\documentclass[11pt]{article}
\usepackage{epsfig}

\usepackage[left=2.5cm,right=2.5cm,top=2.5cm,bottom=2.5cm]{geometry}

\newenvironment{proof}{\noindent {\bf Proof.}}{\bigskip}

\usepackage{epsfig}
\usepackage{amsmath}
\usepackage{amssymb}
\usepackage{color}
\usepackage{subfigure}

\usepackage{tabularx}
\usepackage{epsfig}
\usepackage{amsmath}
\usepackage{amssymb}
\usepackage{graphicx,color, algorithmic, algorithm}
\usepackage{color}
\usepackage{verbatim}
\usepackage{boxedminipage}

{ %
  \newtheorem{lemma}{\textbf{Lemma}}[section]%
  \newtheorem{theorem}[lemma]{\textbf{Theorem}}%
}

\usepackage[cyr]{aeguill}
\usepackage{epsfig}
\usepackage{amsmath}
\usepackage{amssymb}
\usepackage{graphicx,color, algorithmic, algorithm}
\usepackage{color}

\begin{document}

\centerline{\Large\bf Bidirected minimum Manhattan network problem}


\vspace{5mm}

\centerline{\large {\sc N. Catusse, V. Chepoi, K. Nouioua, and Y. Vax\`es}}

\vspace{3mm}

\begin{center}
Laboratoire d'Informatique Fondamentale de Marseille,\\[0.1cm]
Facult\'{e} des Sciences de Luminy, Aix-Marseille Universit\'e,\\[0.1cm]
 F-13288 Marseille Cedex 9, France\\[0.1cm]
{\em\{{catusse,chepoi,nouioua,vaxes\}@lif.univ-mrs.fr}}
\end{center}

\vspace{3mm}

\begin{abstract}
In the {\it bidirected minimum Manhattan network problem}, given  a set $T$ of $n$ terminals in the
plane, we need to construct a network $N(T)$ of minimum total length with the property that the
edges of $N(T)$ are axis-parallel and oriented in a such a way that
every ordered pair  of terminals is connected in $N(T)$ by a directed Manhattan path.
In this paper, we present a polynomial factor 2 approximation algorithm for
the bidirected minimum Manhattan network problem.
\end{abstract}


 \section{Introduction}


 A {\it rectilinear network} $N=(V,E)$  in ${\mathbb R}^2$ consists of a
 finite set $V$ of points and horizontal and vertical
 segments connecting pairs of points of $V.$ The {\it length} of $N$
 is the sum of lengths of its edges. Given a finite set $T$
 of terminals in  ${\mathbb R}^2$, a {\it Manhattan network}
 \cite{GuLeNa} on $T$ is a rectilinear network $N(T)=(V,E)$ such
 that $T\subseteq V$ and for every pair of points in $T,$ the
 network $N(T)$ contains a shortest rectilinear (i.e., Manhattan or $l_1$-) path
 between them. A {\it minimum Manhattan network} on $T$ is a Manhattan network of minimum
 possible length and the Minimum Manhattan Network problem ({\it MMN problem})
 is to find such a network (for an illustration, see Fig.~1). Note that there is always a minimum Manhattan network lying
 in the grid $\Gamma(T)$ defined by the terminals (consisting of all line segments that result
 from intersecting each horizontal and vertical lines through a point in $T$).

In this paper, we consider the following oriented version  of the MMN problem.
In the {\it Bidirected Minimum Manhattan Network problem} (which we abbreviate {\it BDMMN problem}), given a set $T$
 of terminals in the $l_1$-plane,  we want to select a minimum-length subset $N(T)$ of
edges in the grid $\Gamma(T)$ and to orient each edge in $N(T)$
such that each ordered pair of terminals is connected by a directed Manhattan path (for an illustration, see Fig.~2). This oriented version of the minimum Manhattan network problem
was formulated in \cite{GuKlNaSmWo} by M. Benkert and the second author of this note. Further we will assume that $T$ does not contain two terminals
on the same horizontal or vertical line, otherwise the problem does not
have a solution.

The minimum Manhattan network problem has been introduced  by
Gudmundsson, Levcopoulos, and Narasimhan  \cite{GuLeNa}.  Gudmundsson
et al. \cite{GuLeNa}
proposed an $O(n^3)$-time 4-approximation algorithm, and an
$O(n\log{n})$-time 8-approximation algorithm. They also conjectured
that there exists a 2-approximation algorithm and asked if MMN is NP-complete.
Chin, Guo, and Sun \cite{ChGuSu} recently  established that the problem is
indeed NP-complete. Kato, Imai, and Asano \cite{KaImAs} presented
a 2-approximation algorithm, however, their correctness proof is
incomplete (see \cite{BeWoWiSh}). Benkert,
Wolff, Shirabe, and  Widmann \cite{BeWoWiSh} described  an
$O(n\log{n})$-time 3-approximation algorithm and presented a
mixed-integer programming formulation of problem.  Nouioua
\cite{Nou} and Fuchs and Schulze \cite{FuSch} presented two
simple $O(n\log{n})$-time 3-approximation algorithms. The first
correct 2-approximation algorithm (solving the first open
question from \cite{GuLeNa}) was presented by Chepoi, Nouioua, and
Vax\`es \cite{ChNouVa} and is based on a strip-staircase
decomposition and a rounding method applied to the linear program from
\cite{BeWoWiSh}.  In his PhD thesis, Nouioua \cite{Nou}
described an $O(n\log{n})$-time 2-approximation algorithm based on
the primal-dual method.  In 2008, Guo, Sun, and Zhu
\cite{GuSuZh1,GuSuZh} presented two combinatorial factor 2
approximation algorithms with complexity $O(n^2)$ and $O(n\log{n})$ (see also
the PhD thesis \cite{Sch} of Schulze for yet another $O(n\log{n})$-time 2-approximation algorithm).
Seibert and Unger \cite{SeUn} announced a 1.5-approximation
algorithm, however the conference format of their paper does not
permit to understand the description of the algorithm and to check
its claimed performance guarantee (a counterexample that an
important intermediate step is incorrect was
given in \cite{FuSch,Sch}). Finally, a factor 2.5 approximation
algorithm for MMN problem in normed planes with polygonal unit balls was proposed in
\cite{CaChNoVa_normed}.

\begin{figure}
\begin{minipage}[b]{.45\linewidth}
\centering \includegraphics[scale=0.8]{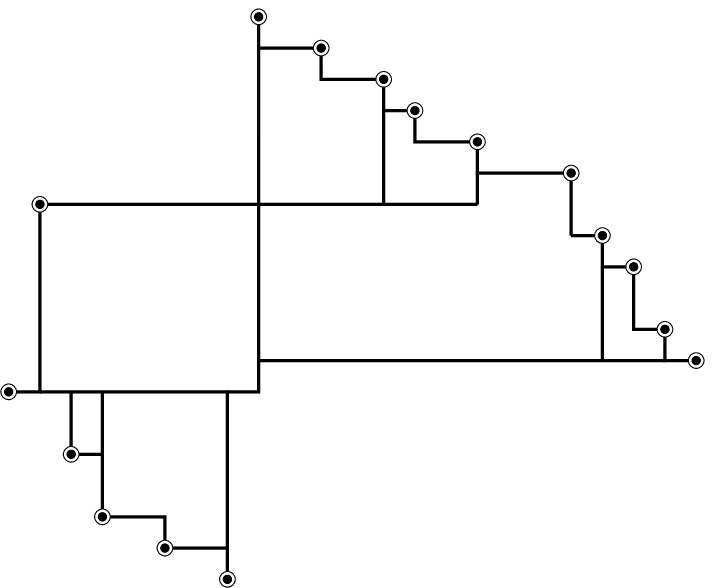}
\caption{\small A minimum Manhattan network\label{man}}
\end{minipage}\hfill\hfill
\begin{minipage}[b]{.55\linewidth}
\centering \includegraphics[scale=0.8]{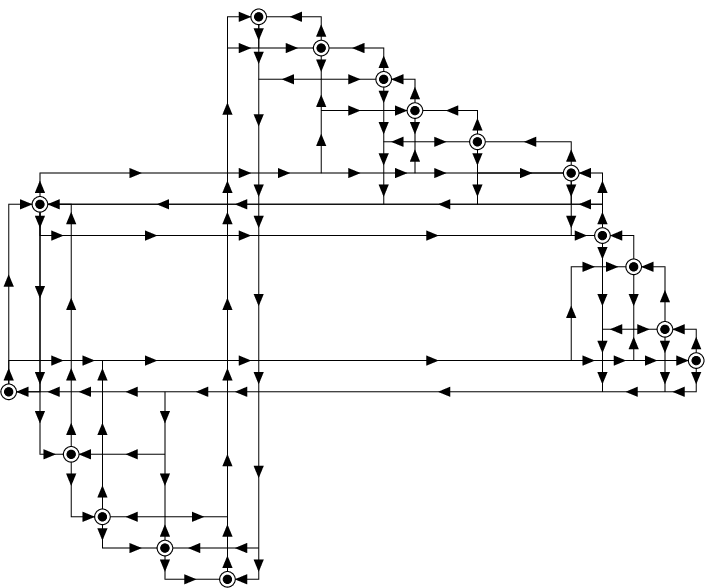}
\caption{\small A minimum bidirected Manhattan network\label{biman}}
\end{minipage}
\end{figure}



Gudmundsson et al. \cite{GuLeNa} introduced the MMN problem in connection with
geometric spanners.  A geometric network $N$ is a {\it $c$-spanner} ($c\ge 1$)
for $T$ if for each pair  ${\bf t}_i,{\bf t}_j\in T,$ there
exists a $({\bf t}_i,{\bf t}_j)$-path in $N$ of length at most $c\cdot \|{\bf t}_i-{\bf t}_j\|.$
In the Euclidean plane,   the unique $1$-spanner of $T$ is the
complete graph on $T.$
In the rectilinear plane, the points are connected by several Manhattan
paths, and the problem of finding the sparsest $1$-spanner
becomes non trivial. Minimum Manhattan networks are precisely the
optimal $1$-spanners. Analogously, the bidirected minimum Manhattan networks can be viewed as  optimal $1$-spanners of the directed grid $\overleftrightarrow{\Gamma}(T)$ obtained from $\Gamma(T)$ by replacing each edge by two opposite directed arcs.
Alternatively, bidirected Manhattan networks are roundtrip $1$-spanners sensu \cite{RoThZw} for the grid $\overleftrightarrow{\Gamma}(T)$ and for the set $T$ of terminals. In both reformulations of bidirected Manhattan networks as directed $1$-spanners of  $\overleftrightarrow{\Gamma}(T)$ we suppose that, like in Manhattan Street Networks \cite{ChAg, Ma, Va}, an edge of $\Gamma(T)$ participating in the resulting spanner can be directed only in one sense.
Geometric spanners have applications in network and VLSI circuit design,
distributed algorithms, and other areas \cite{Epp,NaSm}. Lam,
Alexandersson, and Pachter \cite{LaAlPa} used minimum
Manhattan networks to design efficient search spaces for pair hidden
Markov model alignment algorithms.

In this paper, we propose a factor 2 approximation algorithm for the minimum bidirected
Manhattan network problems. We also solve in the negative Problem 6 from the collection \cite{GuKlNaSmWo} asking
whether a specially constructed bidirected Manhattan network $N_{\varnothing}(T)$ is always optimal.

Our algorithm employs  the strip-staircase decomposition
proposed in our previous paper \cite{ChNouVa} and subsequently used in other approximation algorithms for MMN. First we notice
that each strip, oriented clockwise or counterclockwise, belongs to any bidirected Manhattan network. Then we show that all strips
constituting larger sets, called blocks,  have the same orientation. Since the strips from different blocks do not overlap, one can suppose that the algorithm orients the strips in
the same way as in an optimal solution. Since the bases of a staircase and the strips touching it belong to a common block, it remains to construct in each staircase
a completion of the
strip's solution. Any optimal completion can be subdivided into two subnetworks which, forgetting the orientation, can be viewed as variants of the MMN problem for
staircases. Such optimal (undirected) networks can be computed in polynomial time by dynamic programming. The algorithm then constructs a directed version of these
networks by directing their edges and replacing some edges by two shifted oriented copies.

We conclude this section with some notations. For a point $p \in \mathbb{R}^2$ we will denote by $p^x$ and $p^y$ its two coordinates. For two points $p,q$ of $\mathbb{R}^2$ we will denote by $R(p,q)$ the smallest axis-parallel rectangle containing
$p$ and $q$ (clearly, $p$ and $q$ are two
opposite corners of $R(p,q)$).
Let $T=\{ {\bf t}_1,\ldots,{\bf t}_n\}$ denote the set of
$n$ terminals (recall that $T$ does not contain two terminals on the same vertical or horizontal line).
For two terminals ${\bf t}_i,{\bf t}_j\in T$
we will set $R_{i,j} =R(\bf{t}_i,\bf{t}_j).$
We will say that the rectangle $R_{i,j}$ is {\it empty} if $R_{i,j}\cap T=\{ {\bf t}_i,{\bf t}_j\}.$
Finaly, let $F_{\varnothing}$ be the set of all ordered pairs $(i,j)$ such that $R_{i,j}$ is
empty.

\section{The counterexample}

We start with an example showing that the bidirected network
$N_{\varnothing}(T)$ defined in \cite{GuKlNaSmWo} is not optimal. 
$N_{\varnothing}(T)$ is defined in the following way: go through all empty
rectangles $R_{i,j}$ and orient the edges on the boundary
of $R_{i,j}$ clockwise if the line segment ${\bf t}_i{\bf t}_j$ has
positive slope and counterclockwise if ${\bf t}_i{\bf t}_j$ has negative slope.
$N_{\varnothing}(T)$ is always a bidirected Manhattan network.  In Fig.
\ref{ct_simple} we present an optimal bidirected Manhattan network (its length
is 32) for a set of 5 terminals. For the same set of terminals, the length of the bidirected Manhattan network
$N_{\varnothing}(T)$ presented in Fig. \ref{ct_simple_algo}
is 34:  $N_{\varnothing}(T)$ also includes the two dotted edges of the staircase not
included in the optimal solution. Analogous larger examples show that the ratio between the
length of $N_{\varnothing}(T)$ and the optimum can be arbitrarily large.

\begin{figure}[h]
\begin{minipage}[b]{.45\linewidth}
\centering \includegraphics[width=4cm]{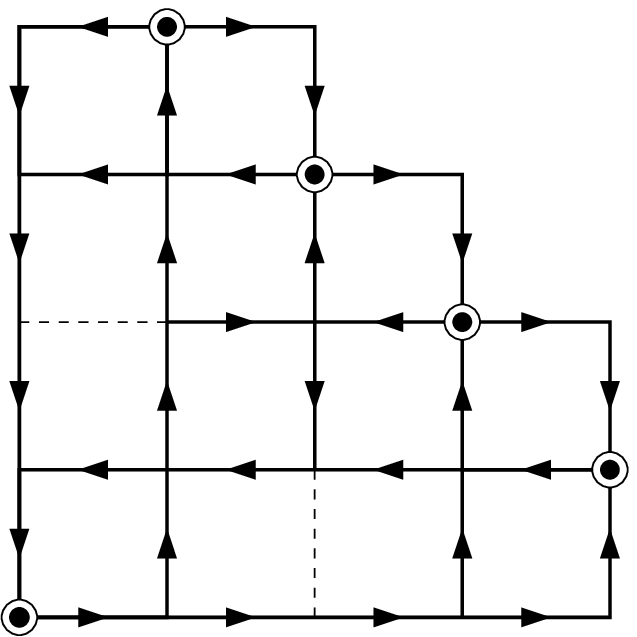}
\caption{An optimal bidirected Manhattan network\label{ct_simple}}
\end{minipage}\hfill
\begin{minipage}[b]{.45\linewidth}
\centering \includegraphics[width=4cm]{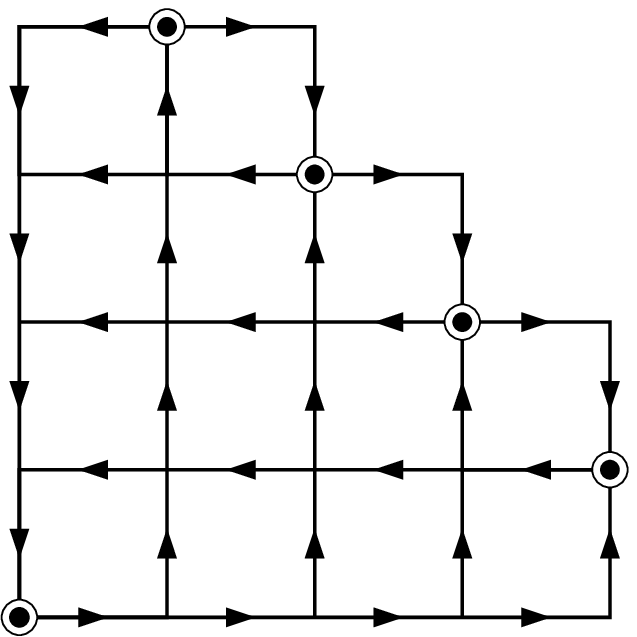}
\caption{The bidirected Manhattan network $N_{\varnothing}(T)$\label{ct_simple_algo}}
\end{minipage}
\end{figure}

\section{Strips and staircases}
In this section, we briefly recall the notions of strips and staircases defined and studied in \cite{ChNouVa}; we refer to this paper
for proofs and some missing details.
An empty
rectangle $R_{i,j}$ is called a {\it vertical strip} if the
$x$-coordinates of ${\bf t}_i$ and ${\bf t}_j$ are consecutive entries of the
sorted list of all $x$-coordinates of the terminals.
Analogously, a empty rectangle $R_{i,j}$ is called a
{\it horizontal strip} if the
$y$-coordinates of ${\bf t}_i$ and ${\bf t}_j$ are consecutive entries of the
sorted list of all $y$-coordinates of the terminals.
The {\it sides} of a vertical (resp., horizontal)
strip $R_{i,j}$ are the vertical (resp., horizontal) sides of
$R_{i,j}.$ Notice that two points ${\bf t}_i, {\bf t}_j$ may define both a
horizontal and a vertical strip. We say that the rectangles
$R_{i,i'}$ and $R_{j,j'}$ form a {\it crossing
configuration} if they intersect and they have the same slope. The
importance of such configurations resides in the following property
whose proof is straightforward:

\begin{lemma} \label{lemma-3.1} If the rectangles $R_{i,i'}$ and $R_{j,j'}$ form a
crossing configuration, then from the two couples of directed $l_1$-paths connecting  ${\bf t}_i$
with ${\bf t}_{i'}$ and  ${\bf t}_j$ with ${\bf t}_{j'}$ one  can derive two couples of directed
$l_1$-paths connecting ${\bf t}_i$ with ${\bf t}_{j'}$ and ${\bf t}_j$ with ${\bf t}_{i'}.$
\end{lemma}

\begin{figure}[h]
 \centering \includegraphics[scale=0.6]{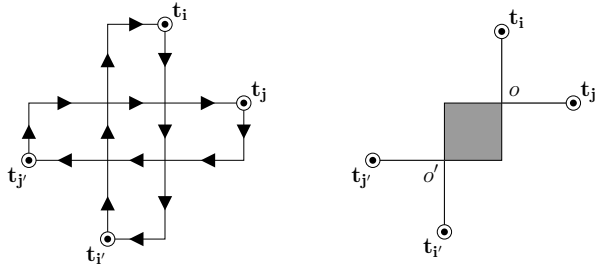}
 \caption{A crossing configuration}
\label{pareto}
\end{figure}

For a crossing
configuration defined by the strips $R_{i,i'},R_{j,j'}$ we can define a rectangle with four tips as illustrated in Fig. \ref{pareto}.
Let $o$ and $o'$ be the two opposite corners of this rectangle labeled in such a way that the four tips
connect $o$ with ${\bf t}_i,{\bf t}_j$ and $o'$ with
${\bf t}_{i'},{\bf t}_{j'}.$ Additionally, suppose without loss of generality,
that ${\bf t}_i$ and ${\bf t}_j$ belong to ${\cal Q}_1(o),$ i.e., to the first
quadrant with respect to the origin $o.$ Then ${\bf t}_{i'}$ and ${\bf t}_{j'}$
belong to ${\cal Q}_3(o').$
Denote by $T_{i,j}$ the set of all
terminals ${\bf t}_k\in (T \setminus \{ {\bf t}_i,{\bf t}_j\}) \cap {\cal Q}_1(o)$
such that (i) $R({\bf t}_k,o)\cap T=\{ {\bf t}_k\}$ and (ii) the region
$\{q\in {\cal Q}_2(o):q^y\le {\bf t}_k^y\}\cup \{q\in {\cal Q}_4(o):q^x\le
{\bf t}_k^x\}$ does not contain any terminal of $T.$ When $T_{i,j}\neq \emptyset,$
we define the {\it staircase} ${\cal S}_{i,j|i',j'}$ as the
union of rectangles $\bigcup \{R(o', {\bf t}_k) : {\bf t}_k \in T_{i,j}\} \setminus R(o, o');$ see Fig.~\ref{staircase} for an illustration. The point
$o$ is called the {\it origin} of this staircase. Analogously one
can define the set $T_{i',j'}$ and the staircase ${\cal
S}_{i',j'|i,j}$ with origin $o'.$ Two other types of staircases will
be defined if ${\bf t}_i,{\bf t}_j$ belong to the second quadrant with respect
to $o$ and ${\bf t}_{i'},{\bf t}_{j'}$ belong to ${\cal Q}_4(o').$ In order to
simplify the presentation, further we will prove all results under
the assumption that the staircase is located in the first quadrant.
By symmetry, all these results also hold for the other types of
staircases.

\begin{figure}[h]
\begin{minipage}[b]{.49\linewidth}
\centering \includegraphics[width=7cm]{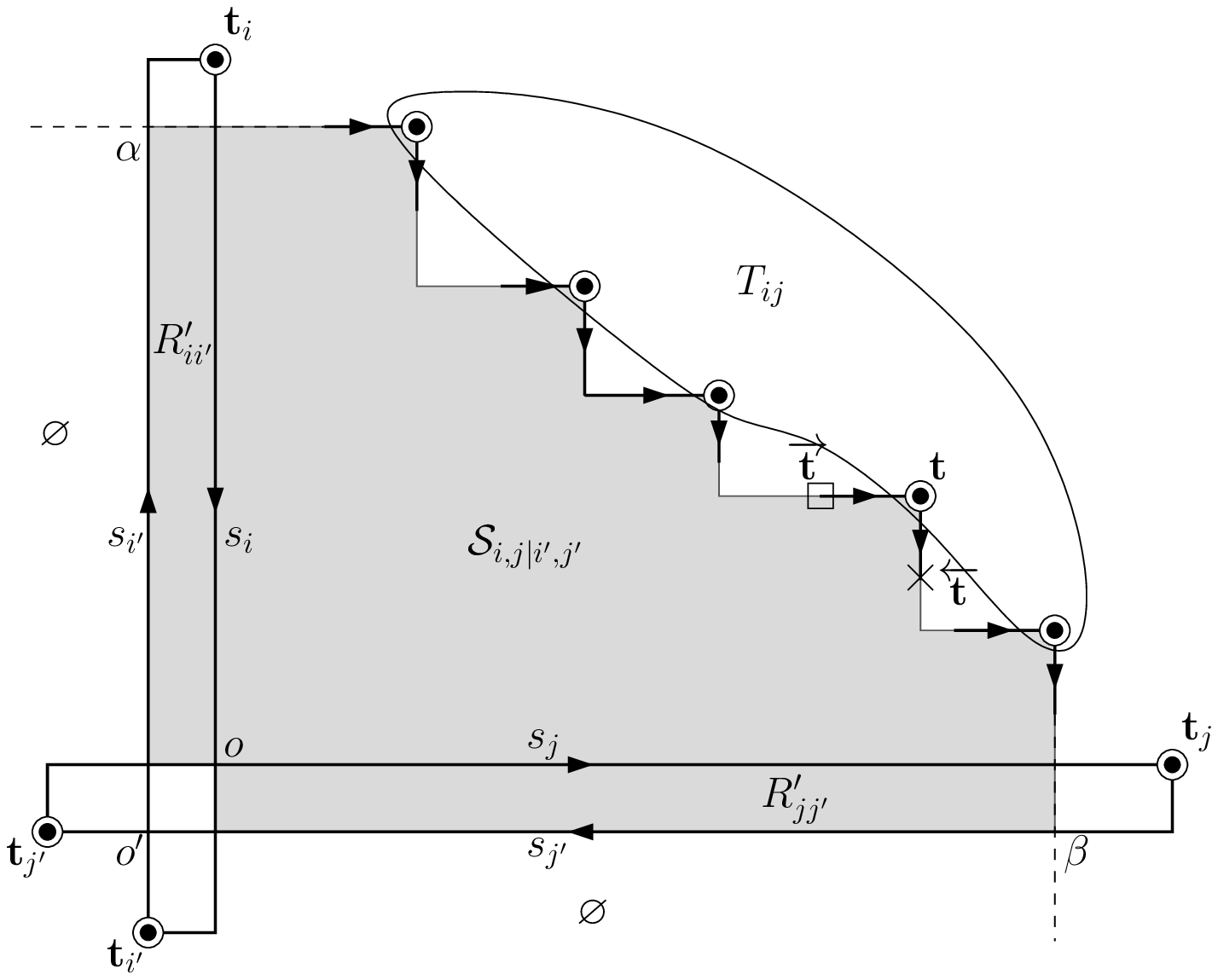}
\caption{Staircase ${\cal S}_{i,j|i',j'}$\label{staircase}}
\end{minipage}\hfill
\begin{minipage}[b]{.49\linewidth}
\centering \includegraphics[width=8cm]{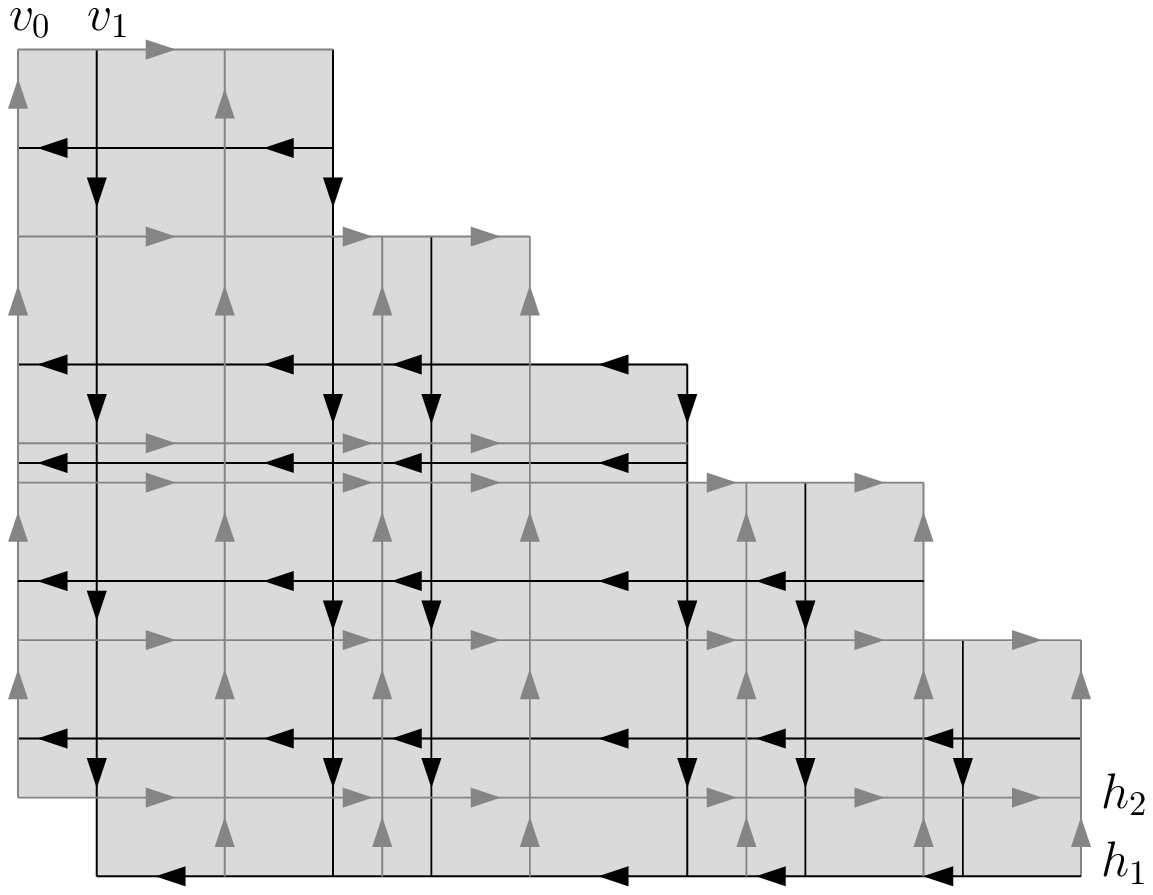}
\caption{${\cal S}_{i,j|i',j'}$ with $\Gamma^{even}_{i,j}$ and $\Gamma^{odd}_{i,j}$ oriented\label{staircase_ori}}
\end{minipage}
\end{figure}

Let $\alpha$ be the leftmost highest point of the staircase ${\cal
S}_{i,j|i',j'}$ and let $\beta$ be the rightmost lowest point of
this staircase.
By definition, ${\cal S}_{i,j|i',j'}\cap T=T_{i,j}.$ By the choice
of $T_{i,j},$ there are no terminals of $T$ located in the regions
$\{q\in {\cal Q}_2(o):q^y\le \alpha^y\}$ and $\{q\in {\cal
Q}_4(o):q^x\le \beta^x\}.$ 


\begin{figure}[h]
\centering \includegraphics[scale=0.8]{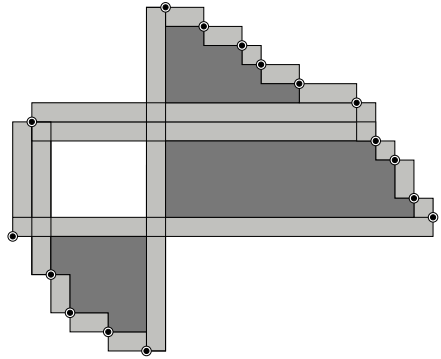} \caption{An example of a strip-staircase decomposition\label{bande_escalier}}
\end{figure}

\section{Blocks and generating sets}

For a strip
$R_{i,j},$  the terminals ${\bf t}_i$ and ${\bf t}_j$ can be connected by two
directed Manhattan paths  only if we take the boundary of the rectangle
$R_{i,j}$ in the solution
and orient it accordingly. Therefore the boundary of each strip $R_{i,j}$ belongs to all bidirected Manhattan networks. $R_{i,j}$ has only two orientations
(clockwise and counterclockwise) producing the two directed Manhattan paths between ${\bf t}_i$ and ${\bf t}_j$. We say that (the orientations of) two rectangles $R_{i,j}$ and $R_{i',j'}$ are {\it compatible} if they have the same slope and are oriented in the same way or if they have different slopes and are oriented in opposite ways.
Clearly, two strips sharing an edge $e$ of $\Gamma(T)$ must be compatible.
We extend this property to larger sets, called blocks.

Let $P\subseteq T$ be a maximal by inclusion set of terminals such that for all  ${\bf t}_i\in P$ (the same) two opposite quadrants centered at ${\bf t}_i$ are empty, i.e., their
intersections with $T$ consist only of ${\bf t}_i;$ suppose that these empty quadrants are
the second and the fourth quadrants $Q_2({\bf t}_i)$ and $Q_4({\bf t}_i)$. Now, let the points of $P$  be sorted by $x$-coordinate. The $i$th {\it block}  is the set of all terminals of
$T$ contained in the axis-parallel rectangle spanned by $i$th and $(i+1)$th points of $P,$ the first block consists of all terminals located in the third quadrant
defined by the first point of $P,$ and the last block  consists of all terminals located in the first quadrant defined by the last point of $P.$ From the definition
follows that two terminals defining a strip belong to a common block (which some abuse of language, we will say that the strip itself belongs to this block).

\begin{lemma} \label{orientation_strips} In any bidirected Manhattan network $N(T)$, all strips from the same block $B$ are compatible.
\end{lemma}

\begin{proof} Consider a graph whose vertices are the strips from $B$ and two strips are adjacent if and only if their boundaries share an edge of $\Gamma(T)$ or a terminal. This graph is connected because  its two subgraphs induced by horizontal and vertical strips are connected and any terminal of $B$ defines in $B$ a vertical and a horizontal strip which are adjacent in this graph.
Therefore, if $B$ contains incompatible strips, then we can find in this graph two adjacent incompatible strips  $R_{i,j}$ and  $R_{j',k}$. Since two strips sharing an edge are compatible, $R_{i,j}$ and $R_{j',k}$ necessarily share a terminal, say $j=j'.$
We can suppose without loss of generality that ${\bf t }_i \in Q_1({\bf t}_j)$ and ${\bf t}_k \in Q_3({\bf t}_j),$ otherwise the boundaries of these strips overlap. Hence  ${\bf t}_j \notin P,$ and $Q_2({\bf t}_j)$ or $Q_4({\bf t}_j)$ is not empty. Suppose that there is a point ${\bf t}_l \in Q_4({\bf t}_j)$ (see Fig. \ref{block} for an illustration).
Since $R_{i,j}$ and $R_{j,k}$ are incompatible, ${\bf t}_j$ is the head or the tail of both edges incident to ${\bf t}_j$ in $Q_4({\bf t}_j),$ say the tail.
 Since ${\bf t}_l \in Q_4({\bf t}_j),$ the rectangle $R_{j,l}$ is also included in $Q_4({\bf t}_j).$ Since any (directed or not) Manhattan path between ${\bf t}_j$ and ${\bf t}_l$ is included in $R_{j,l}$  and both edges of $N(T)\cap R_{j,l}$ incident to ${\bf t}_j$ are directed away from ${\bf t}_j$, we will not be able to produce a directed Manhattan path from ${\bf t}_l$ to ${\bf t}_j$ in $N(T),$ a contradiction. $\Box$
\end{proof}

\begin{figure}[h]
\label{block}
\centering \includegraphics[scale=0.6]{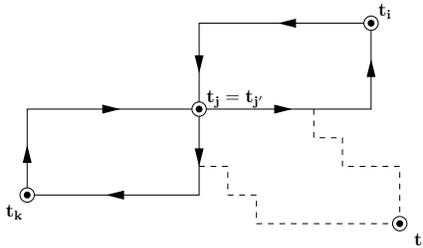}
\caption{To the proof of Lemma \ref{orientation_strips}.}
\end{figure}

%

We continue by
adapting to BDMMN the notion of a generating set introduced in \cite{KaImAs}
for MMN problem: a {\it generating set}  is a subset $F$ of ordered pairs of terminals of $T$ with
the property that a directed subnetwork of $\Gamma(T)$ containing directed Manhattan paths between
all pairs of $F$ is a bidirected Manhattan network for $T.$  Let $F'$ be the set of all  ordered pairs of
terminals defining strips. Let also $F''$ be the set of all ordered pairs $({\bf t}_{j'},{\bf t}_l)$ and $({\bf t}_{l},{\bf t}_{j'})$
such that there exists a staircase ${\cal S}_{i,j|i',j'}$ with
${\bf t}_l$ belonging to the set $T_{i,j}$ of all terminals defining the corners of ${\cal S}_{i,j|i',j'}$. The proof of the following result closely
follows the proof of Lemma 3.3 of  \cite{ChNouVa}.

\begin{lemma} \label{generating_set_oriented} $F:=F'\cup F''$ is a generating set.
\end{lemma}

\begin{proof} The set $F_{\varnothing}$ of empty pairs is clearly a generating set.
Let $N$ be a bidirected rectilinear network containing directed $l_1$-paths
for all ordered pairs in $F.$ To prove that $N$ is a bidirected Manhattan network on
$T$, it suffices to establish that for  any  arbitrary pairs $(k,k'),(k',k)\in F_{\varnothing}\setminus F,$ in $N$ there exists a directed Manhattan
path from ${\bf t}_k$ to
${\bf t}_{k'}$ and a directed Manhattan path from ${\bf t}_{k'}$ to ${\bf t}_k$.  Assume without loss of
generality that ${\bf t}^x_{k'}\le {\bf t}^x_k$ and ${\bf t}^y_{k'}\le {\bf t}^y_k.$ Since
$(k,k')\in F_{\varnothing},$ the rectangle $R_{k,k'}$ is empty.
The vertical and horizontal lines through the points ${\bf t}_k$ and
${\bf t}_{k'}$ partition the plane into the rectangle $R_{k,k'},$ four
open quadrants and four closed unbounded half-bands labeled
counterclockwise ${\cal H}_1,{\cal H}_2,{\cal H}_3,$ and ${\cal
H}_4$ (see Fig. \ref{lemme_generateur}). Consider the leftmost terminal ${\bf t}_{i_1}$ of
${\cal H}_1$ (this
terminal exists because ${\bf t}_k\in {\cal H}_1$). Now, consider the
rightmost terminal ${\bf t}_{i'_1}$ of ${\cal H}_3$ such that
${\bf t}_{i'_1}^x\le {\bf t}_{i_1}^x$ (again this terminal exists because ${\bf t}_{k'}\in {\cal
H}_3$ and ${\bf t}_{k'}^x \le {\bf t}_{i_1}^x$). By the choice of ${\bf t}_{i_1}$ and
${\bf t}_{i'_1},$ the rectangle $R_{i_1,i'_1}$ is the leftmost vertical
strip crossing the rectangle $R_{k,k'}.$ Analogously, by letting
${\bf t}_{i_2}, {\bf t}_{j_1},$ and ${\bf t}_{j_2}$ be the rightmost terminal of
${\cal H}_3,$ the lowest
terminal of ${\cal H}_4$ and the highest terminal of ${\cal H}_2,$
respectively, we obtain the rightmost  vertical strip
$R_{i_2,i'_2},$ the lowest horizontal strip $R_{j_1,j'_1},$ and the
highest horizontal strip $R_{j_2,j'_2}$ crossing the rectangle
$R_{k,k'}.$ Notice that the strips $R_{j_2,j'_2}$ and $R_{i_2,i'_2}$
as well as the strips $R_{j_1,j'_1}$ and $R_{i_1,i'_1}$ constitute
crossing configurations.

\begin{figure}
\begin{center}
\centering\includegraphics[scale=0.9]{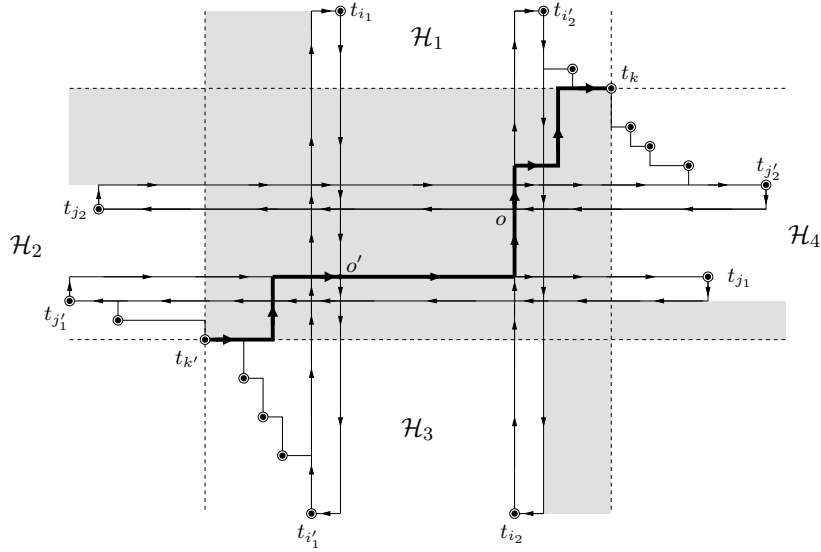}
\caption{To the proof of Lemma \ref{generating_set_oriented}}
\label{lemme_generateur}
\end{center}
\end{figure}

Now, we will prove that $N$ contains a directed  $l_1$-path from ${\bf t}_{i_2}$ to
${\bf t}_k$ and a directed $l_1$-path from ${\bf t}_{k'}$ and ${\bf t}_{j_1}.$ We
distinguish three cases. If ${\bf t}_k={\bf t}_{i'_2},$ then
$R_{i_2,k}=R_{i_2,i'_2}$ is a strip and thus $(k,i_2)\in F.$ If
${\bf t}_k={\bf t}_{j'_2},$ then the strips $R_{j_2,k}$ and $R_{i_2,i'_2}$ form
a crossing configuration. By Lemma \ref{lemma-3.1}, from the directed
$l_1$-paths of $N$ running from ${\bf t}_{j_2}$ to ${\bf t}_{k}$ and from
${\bf t}_{i_2}$ to ${\bf t}_{i'_2},$ we can derive a couple of directed $l_1$-paths connecting ${\bf t}_k$
with ${\bf t}_{i_2}.$ Finally, if ${\bf t}_k\notin \{{\bf t}_{i'_2},{\bf t}_{j'_2}\},$ we
assert that the crossing configuration $R_{i_2,i'_2}$ and
$R_{j_2,j'_2}$ defines a staircase ${\cal S}_{i'_2,j'_2|i_2,j_2}$ such
that ${\bf t}_k$ belongs to $T_{i'_2,j'_2}.$ Indeed, let $o$ be the
highest leftmost intersection point of the strips $R_{i_2,i'_2}$ and
$R_{j_2,j'_2}$ (see Fig. \ref{lemme_generateur}). Since $R({\bf t}_k,o)$ is contained in the
empty rectangle $R_{k,k'},$ we conclude that $R({\bf t}_k,o) \cap T
= \{{\bf t}_k\}.$ Moreover, by the choice of ${\bf t}_{i_2}$ and ${\bf t}_{j_2},$ the
unbounded half-bands $\{q\in {\cal H}_3 : q^x\ge {\bf t}_{i'_2}\}$ and
$\{q\in {\cal H}_2 : q^y \ge {\bf t}_{j'_2}\}$ do not contain terminals
(in Fig.~\ref{lemme_generateur}, the shaded region does not contain terminals), thus
establishing our assertion. This implies that ${\bf t}_k\in
T_{i'_2,j'_2},$ whence $(k,i_2),(i_2,k)\in F.$ Therefore, in all three
cases the terminals ${\bf t}_k$ and ${\bf t}_{i_2}$ are connected in $N$ by a couple of directed
$l_1$-paths. Using a similar analysis, one can show that ${\bf t}_{k'}$ and
${\bf t}_{j_1}$ are also connected in $N$ by a couple of directed $l_1$-paths. By
construction, the rectangles $R_{k,i_2}$ and $R_{k',j_1}$ form a
crossing configuration and thus, by Lemma \ref{lemma-3.1}, there is
a couple of $l_1$-paths of $N$ between the terminals ${\bf t}_k$ and ${\bf t}_{k'},$
concluding the proof. $\Box$

\end{proof}

For a staircase ${\cal S}_{i,j|i',j'},$ let $T^+_{i,j}$ be the set consisting of $T_{i,j},$ the four terminals ${\bf t}_i,{\bf t}_j,{\bf t}_{i'},{\bf t}_{j'}$ of the bases of ${\cal S}_{i,j|i',j'}$, and the terminals defining strips touching the boundary of ${\cal S}_{i,j|i',j'}.$

\begin{lemma} \label{one_block} $T^+_{i,j}$ is included in a block.
\end{lemma}

\begin{proof} Suppose by way of contradiction that there exists  ${\bf t}_k\in P$ such that two terminals ${\bf t}_l$ and ${\bf t}_m$ of  $T^+_{i,j}$ belong to different quadrants $Q_1({\bf t}_k)$
and $Q_3({\bf t}_k)$ (and the quadrants $Q_2({\bf t}_k)$ and $Q_4({\bf t}_k)$ are empty). Since the interiors of  ${\cal S}_{i,j|i',j'},$ of  the rectangles $R_{i,i'},R_{j,j'},$ and of the  regions ${\mathcal Q}',{\mathcal Q}''$ defined in previous section are all empty, ${\bf t}_k$ can be located only in the first quadrant defined by a concave vertex of ${\cal S}_{i,j|i',j'}.$ But in this
case, we can find two terminals of ${\cal S}_{i,j|i',j'}$ located in $Q_2({\bf t}_k)$ and $Q_4({\bf t}_k),$ contrary to the choice of  ${\bf t}_k$ in $P$. $\Box$
\end{proof}

\section{The algorithm}
By Lemma \ref{orientation_strips}, all strips of any block are compatible. Since the strips from different blocks are  edge-disjoint, the algorithm can test the two possible orientations of each block independently of the orientations of other blocks. Thus, we can suppose that the algorithm oriented the strips of each block in the same way as in an optimal bidirected Manhattan network $N^*(T).$  Let $N'(T)$ be the union of
all boundaries of  strips directed this way. By Lemma \ref{one_block}, the bases of a staircase  ${\cal S}_{i,j|i',j'}$ and the strips touching ${\cal S}_{i,j|i',j'}$ belong to the same block $B$, therefore they are all compatible and their orientation can be assumed fixed. Since the bases of ${\cal S}_{i,j|i',j'}$ have the same slope, they are oriented both clockwise or both counterclockwise.
From \cite{ChNouVa} we know that any strip may touch the boundary of a staircase but cannot cross its interior.  Let $N'_{i,j}$ be the intersection of ${\cal S}_{i,j|i',j'}$ with the union of the boundaries of the strips
from $B,$ i.e., $N'_{i,j}=N'(T)\cap {\cal S}_{i,j|i',j'}.$ Hence, by Lemma \ref{generating_set_oriented} it remains, for each staircase ${\cal S}_{i,j|i',j'},$ to complete $N'_{i,j}$ to a local bidirected network $N''_{i,j,}$ such that any pair $({\bf t}_{j'},{\bf t}_l),({\bf t}_{l},{\bf t}_{j'})$ with ${\bf t}_{l}\in T_{i,j}$ can be connected in $N''_{i,j} \cup N'(T)$ by a directed Manhattan path. This must be done in such a way that the length of the network $N_{i,j}=N''_{i,j}\setminus N'_{i,j}$ is as small as possible. Let $N^*_{i,j}$ be an optimal completion of $N'_{i,j}.$

Suppose that $R_{i,i'}$ and $R_{j,j'}$ are the vertical and the horizontal bases
of  ${\cal S}_{i,j|i',j'}$ (see Fig.~\ref{staircase}). Let $R'_{i,i'}=R_{i,i'}\cap {\cal S}_{i,j|i',j'}$ and $R'_{j,j'}=R_{j,j'}\cap {\cal S}_{i,j|i',j'}.$  Suppose
that in algorithm's
and optimal solutions, these strips are oriented clockwise. Hence the leftmost
vertical side $s_{i'}$ of $R'_{i,i'}$ is upward, the  opposite side $s_i$ is
downward, the upper horizontal side $s_j$ of  $R'_{j,j'}$ is to the right, and
its opposite side $s_{j'}$ is to the left. For a terminal ${\bf t}\in T_{i,j},$
denote by $\overrightarrow{\bf t}$ and $\overleftarrow{\bf t}$ the tail and the
head of the directed edges of $N'_{i,j}$ incident to ${\bf t}$.
Set $\overrightarrow{T}_{i,j}=\{ \overrightarrow{\bf t}: {\bf t}\in T_{i,j}\}$
and $\overleftarrow{T}_{i,j}=\{ \overleftarrow{\bf t}: {\bf t}\in T_{i,j}\}$
(they are all vertices of the grid $\Gamma (T)$).

Any optimal completion $N^*_{i,j}$ of $N'_{i,j}$ can be decomposed into two
edge-disjoint subnetworks $\overrightarrow{N^*}_{i,j}$ and $\overleftarrow{N^*}_{i,j},$ such that $\overrightarrow{N^*}_{i,j}$ contains the edges on the directed Manhattan paths running from $s_{i'}\cup s_j$ to the points of $\overrightarrow{T}_{i,j}$ and  $\overleftarrow{N^*}_{i,j}$ contains the edges on the directed Manhattan paths running from the points of $\overleftarrow{T}_{i,j}$ to $s_i\cup s_{j'}.$ The length of $\overrightarrow{N^*}_{i,j}$ cannot be smaller than the length of an optimal (non-oriented) network $A_{i,j}$ in $\Gamma(T)$ connecting the points of $\overrightarrow{T}_{i,j}$ to $s_{i'}\cup s_j$ by Manhattan paths. Analogously, the length of $\overleftarrow{N^*}_{i,j}$ cannot be smaller than the length of an optimal (non-oriented) network $B_{i,j}$ in $\Gamma(T)$ connecting the points of $\overleftarrow{T}_{i,j}$ to $s_{i}\cup s_{j'}.$ At the difference of $\overrightarrow{N^*}_{i,j}$ and $\overleftarrow{N^*}_{i,j},$ the networks $A_{i,j}$ and $B_{i,j}$ are not edge-disjoint.
However, we can compute optimal $A_{i,j}$ and $B_{i,j}$ in polynomial time using dynamic programming because each of the sets of terminals $\overrightarrow{T}_{i,j}$ and $\overleftarrow{T}_{i,j}$ also constitute a staircase and to compute an optimal solution we will have to solve only a polynomial number of subproblems (this problem is similar to Steiner arborescence or to MMN problems for terminals on a staircase).

The algorithm computes by dynamic programming an optimal network $A_{i,j}$  for connecting $\overrightarrow{T}_{i,j}$ to $s_{i'}\cup s_j$ and an optimal network $B_{i,j}$
for connecting $\overleftarrow{T}_{i,j}$ to $s_{i}\cup s_{j'}.$  The dynamic programming constructs $A_{i,j}$ recursively in the following way: there exists a point
$\overrightarrow{\bf t}\in \overrightarrow{T}_{i,j}$ which is either connected in $A_{i,j}$  to $s_{i'}$ by a horizontal segment $s'$ or to $s_j$ by a vertical segment $s''$, say the first.
Then the problem is subdivided into two smaller subproblems, one for the points of $\overrightarrow{T}_{i,j}$ located strictly above $s'$  and the union $s_{i'}\cup s'$ and another for $s_{i'}\cup s_j$ and the
points of $\overrightarrow{T}_{i,j}$ located strictly below $s'$ (the case when $\overrightarrow{\bf t}$ is connected vertically is analogous). The construction of $B_{i,j}$ is similar (see the first two networks in Fig. \ref{AB_ij} for an illustration of $A_{i,j}$ and $B_{i,j}$).


Finally,  the algorithm ``rounds'' each of the networks $A_{i,j}$ and $B_{i,j}$  in order to produce directed networks
$\overrightarrow{A}_{i,j}$ and $\overleftarrow{B}_{i,j}$ having lengths at most twice the lengths of  $A_{i,j}$ and $B_{i,j},$ respectively (see the last two networks from Fig. \ref{AB_ij}).
The algorithm returns $N_{i,j}=\overrightarrow{A}_{i,j}\cup \overleftarrow{B}_{i,j}$ as a local completion of $N'_{i,j}.$ The networks
$\overrightarrow{A}_{i,j}$ and $\overleftarrow{B}_{i,j}$ are constructed in the following way. Let $v_0,v_1,\ldots,v_{k-1}$ and $h_1,h_2,\ldots,h_{l}$ be the consecutive horizontal and vertical lines of the grid $\Gamma(T)$ intersecting the staircase ${\cal S}_{i,j|i',j'}$ and numbered in such a way that $s_{i'}\subset v_0, s_i\subset v_1$ and $s_{j'}\subset h_1, s_j\subset h_2$. (If the bases are oriented counterclockwise, then we consider the same lines but we index them $v_1,v_2,\ldots,v_{k}$ and $h_0,h_1,\ldots,h_{l-1}$.)
Let $\Gamma^{even}_{i,j}$ (respectively, $\Gamma^{odd}_{i,j}$) be the
subgrid of $\Gamma(T)\cap {\cal S}_{i,j|i',j'}$ induced by vertical
and horizontal lines with even indices (respectively, with odd indices).
(Notice that if we orient the horizontal edges of $\Gamma^{even}_{i,j}$ to the right and the vertical edges upward,
and the horizontal edges of $\Gamma^{odd}_{i,j}$ to the left and the vertical edges downward, then we obtain a network which is well-known in the literature  as a {\it Manhattan Street Network} (see Fig. \ref{staircase_ori}) \cite{ChAg, Ma, Va}.)


\begin{figure}[h]\label{AB_ij}
\centering
\subfigure[$A_{i,j}$]{
\includegraphics[width=6cm]{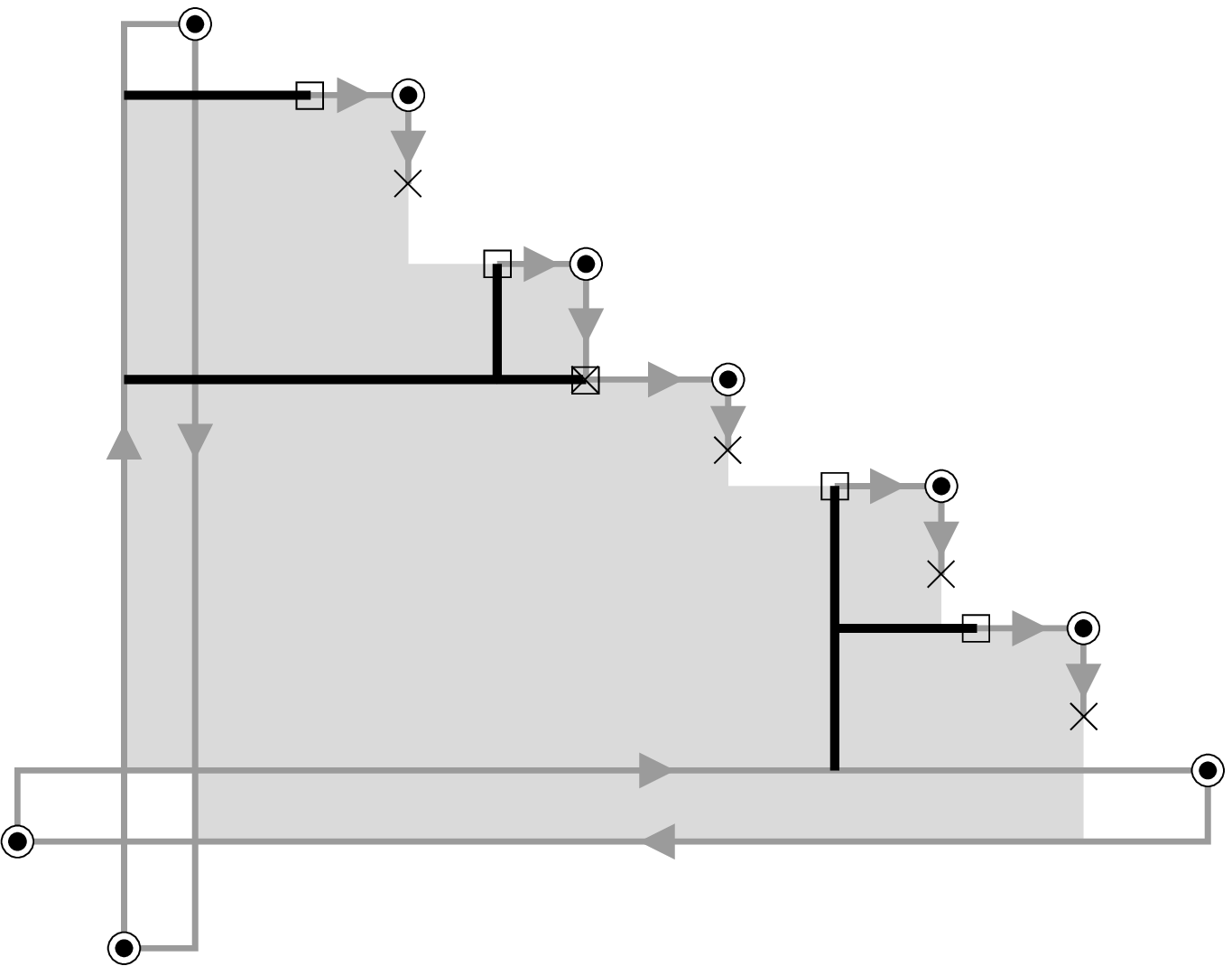}
\label{subfig1}
}\hspace{2cm}
\subfigure[$B_{i,j}$]{
\includegraphics[width=6cm]{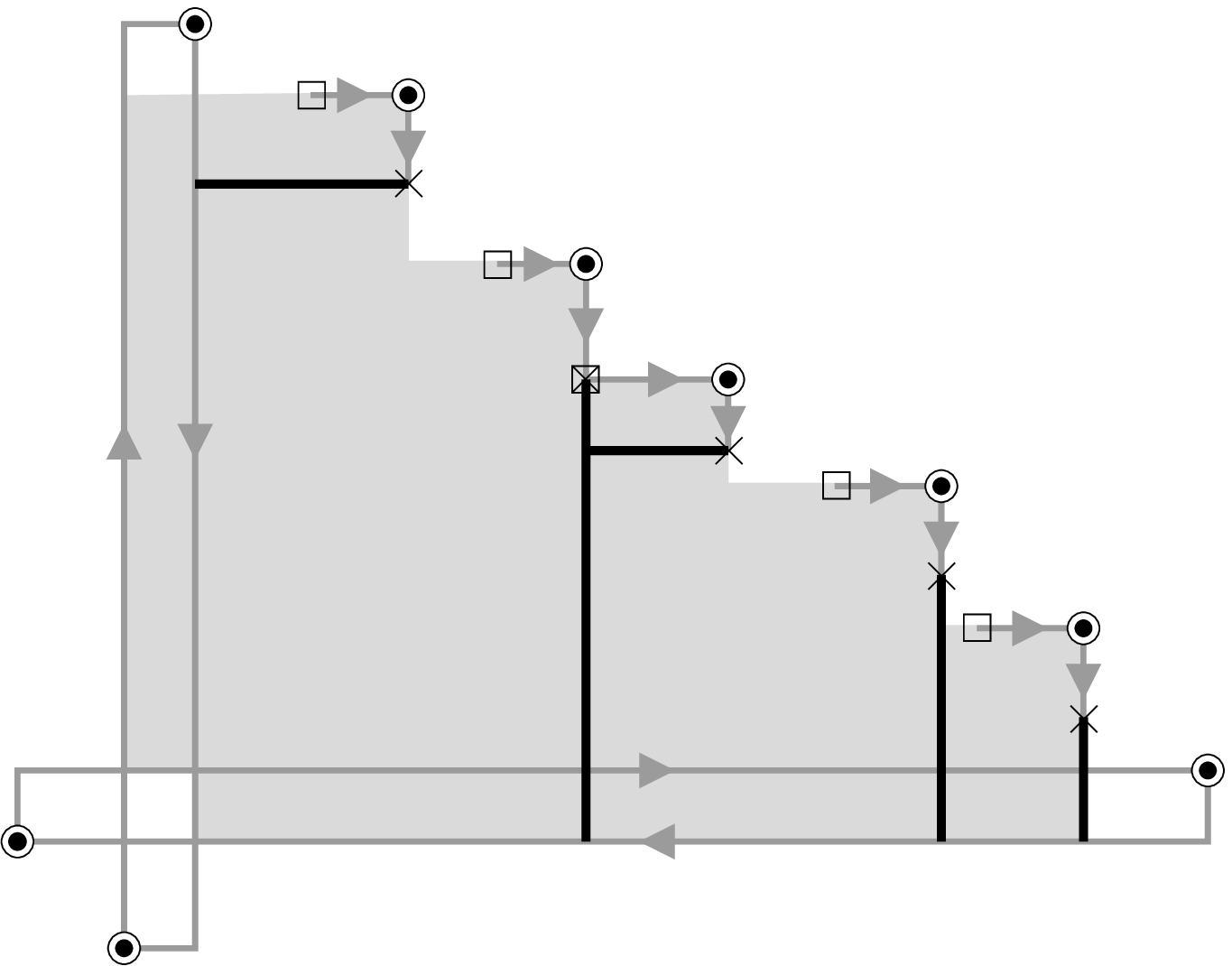}
\label{subfig2}
}
\subfigure[$\protect\overrightarrow{A}_{i,j}$]{
\includegraphics[width=6cm]{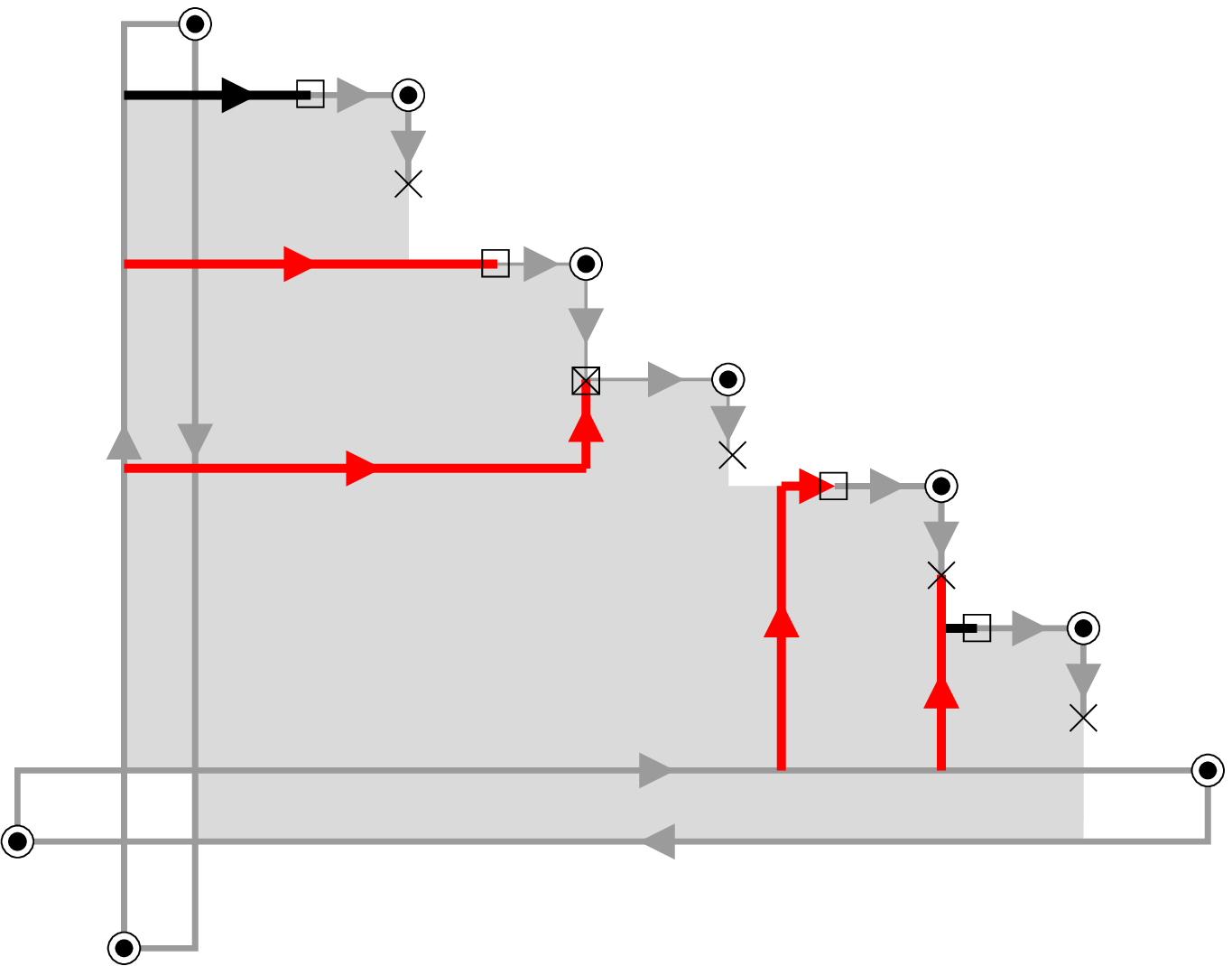}
\label{subfig3}
}\hspace{2cm}
\subfigure[$\protect\overleftarrow{B}_{i,j}$]{
\includegraphics[width=6cm]{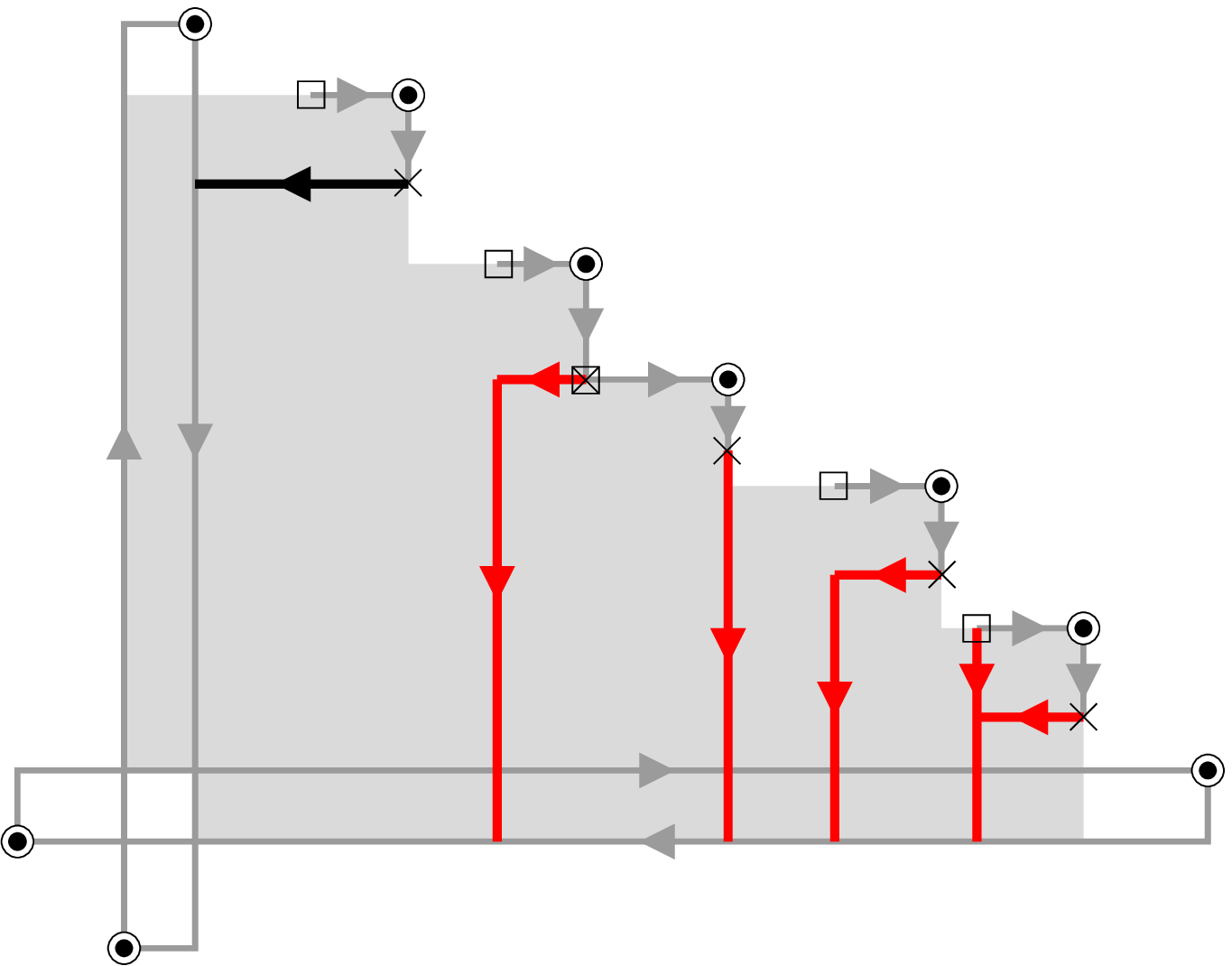}
\label{subfig4}
}
\label{networks}
\caption{The networks $A_{i,j}, B_{i,j}$ and the rounded directed networks $\protect\overrightarrow{A}_{i,j},\protect\overleftarrow{B}_{i,j}$ }
\end{figure}

\begin{figure}[h]
\centering \includegraphics[width=6cm]{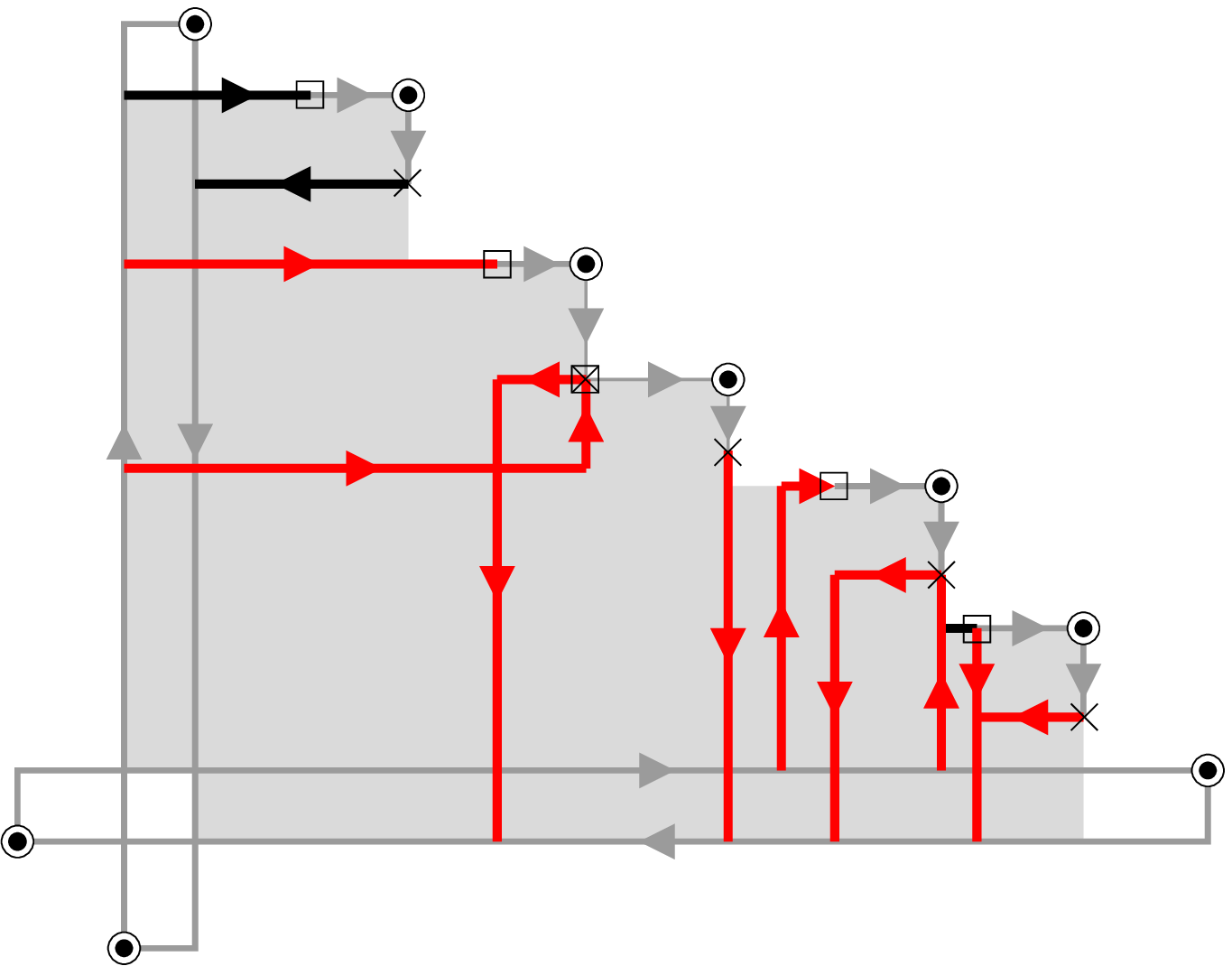}
\caption{Network $\protect\overrightarrow{A}_{i,j}\cup \protect\overleftarrow{B}_{i,j}$}
\end{figure}

The algorithm recursively derives $\overrightarrow{A}_{i,j}$ from $A_{i,j}$ and
$\overleftarrow{B}_{i,j}$ from $B_{i,j}.$ We show how to construct
$\overrightarrow{A}_{i,j}$, but each step of the algorithm is performed for both
 $\overrightarrow{A}_{i,j}$ and $\overleftarrow{B}_{i,j}$ (before going to the
next step). First, in Step 1, we insert in $\overrightarrow{A}_{i,j}$ the edges
of $\Gamma^{even}_{i,j}$ which have their
support in $A_{i,j}$ and orient them upward or to the right. The remaining directed edges are added in order in which the segments of $A_{i,j}$  have been added by the dynamic programming algorithm.
For a current  $\overrightarrow{{\bf t}},$ let $s$ be the vertical or horizontal segment of $A_{i,j}$ connecting  $\overrightarrow{{\bf t}}$
to the previously defined part of $\overrightarrow{A}_{i,j}.$ If $s$ belongs to a horizontal or vertical line with odd index $m,$ say $s$ belongs to $h_m,$ then in Step 2 we include
in $\overrightarrow{A}_{i,j}$ the segments $s',s''$ oriented to the right which
correspond to paths of $\Gamma^{even}_{i,j}$ obtained by
intersecting $h_{m-1}$ and $h_{m+1}$ with the vertical
lines passing via the ends of $s.$ If $s_j\subset h_{m-1}$ (i.e., $m=3$), then
$s'$ is not added. Additionally, we remove from $\overrightarrow{A}_{i,j}$ all
vertical edges whose lowest end-vertex is comprised between $s'$ and $s''$ (by
the construction, this pruning operation removes the edges that are no longer
used by directed Manhattan paths in $\overrightarrow{A}_{i,j}$).  In Step 3, we
proceed the points of $\overrightarrow{T}_{i,j}\cap \Gamma^{odd}_{i,j}$ in
the same order as in Step 2. Let $c$ be the vertical segment connecting
$\overrightarrow{{\bf t}}$ to $h_{m-1}.$ If  $c$ does not belong to
$\overleftarrow{B}_{i,j}$ after Step 2 (Fig. \ref{step2_3}, Step 3 (a)), we add $c$ oriented upward to
$\overrightarrow{A}_{i,j}.$ Otherwise, we consider the cell of $\Gamma(T)$ whose
boundary contains $c$ and a subsegment $b'$ of $s'$ (Fig. \ref{step2_3}, Step 3 (b)), we remove $b'$ and add to $\overrightarrow{A}_{i,j}$ the alternative path
around this cell consisting of an upward twin $c'$ of $c$ and a twin
$b$ of $b'$ oriented to the right ($b$ is a subsegment of $s$).
We call such a path a {\it replacement path}.
Let $N''_{i,j}:=N'_{i,j}\cup \overrightarrow{A}_{i,j}\cup \overleftarrow{B}_{i,j}$ and let $l(N''_{i,j})$ be its length.

\begin{figure}[h]
 \centering \includegraphics[scale=0.8]{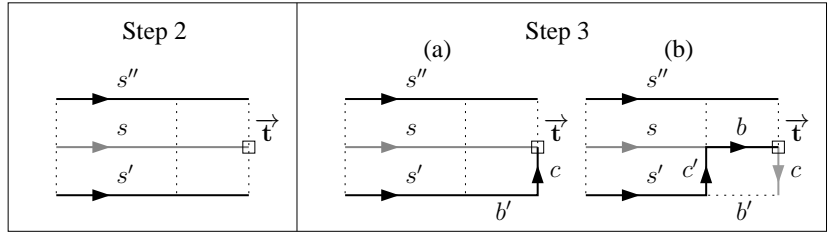}
 \caption{Steps 2 and 3 of Algorithm 1}
\label{step2_3}
\end{figure}

We conclude this section with the pseudocode of the algorithm constructing the directed networks
$\protect\overrightarrow{A}_{i,j}$ and $\protect\overleftarrow{B}_{i,j}$ from  $A_{i,j}$ and $B_{i,j}$:
\begin{algorithm}[h!]
\begin{algorithmic}
\caption{{\sf Construction of $\protect\overrightarrow{A}_{i,j} \cup \protect\overleftarrow{B}_{i,j}$}}
\REQUIRE{Networks $A_{i,j}$ and $B_{i,j}.$}
\STATE {\bf Step 1:} Insert in $\overrightarrow{A}_{i,j}$ the edges
of $\Gamma^{even}_{i,j}$ which have their
support in $A_{i,j}$ and orient them upward or to the right.
\STATE {\bf Step 2:} For each remaining directed edge (in order in which the segments of $A_{i,j}$  have been added by the dynamic programming algorithm):
\STATE \hspace{0.5cm}  If $s_j\not\subset h_{m-1}$ (i.e., $m \neq 3$), then include in $\overrightarrow{A}_{i,j}$ the segment $s'$ oriented to the right.
\STATE \hspace{0.5cm}  Include in $\overrightarrow{A}_{i,j}$ the segment $s''$ oriented to the right.
\STATE Remove from $\overrightarrow{A}_{i,j}$ all
vertical edges whose lowest end-vertex is comprised between $s'$ and $s''.$
\STATE Perform the Steps 1-2 for computing $\protect\overleftarrow{B}_{i,j}$ from $B_{i,j}.$
\STATE {\bf Step 3:} For each point of $\overrightarrow{T}_{i,j}\cap \Gamma^{odd}_{i,j}$ (proceeded in
the same order as in Step 2):
\STATE \hspace{0.5cm}  If $c$ does not belong to $\overleftarrow{B}_{i,j},$ then add to
$\overrightarrow{A}_{i,j}$ the segment $c$ oriented upward.
\STATE \hspace{0.5cm}   Otherwise, remove $b'$ and add to $\overrightarrow{A}_{i,j}$ the alternative path
consisting of $c'$ oriented upward and $b$ oriented to the right (see Fig. \ref{step2_3}).
\STATE 
Perform the Step 3 for each point of $\overleftarrow{T}_{i,j}\cap \Gamma^{even}_{i,j}.$
\end{algorithmic}
\end{algorithm}

\section{The analysis of the algorithm}

Now, we will show that the algorithm described in previous section returns a bidirected Manhattan network and that the length of this network
is at most twice the length of an optimal bidirected Manhattan network.

\begin{lemma} \label{directed_grids}  The supports of $\overrightarrow{A}_{i,j}$ and $\overleftarrow{B}_{i,j}$ are disjoint.
Moreover, in $\overrightarrow{A}_{i,j}$ there exists a directed Manhattan path from $s_{i'}\cup s_j$ to each point of $\overrightarrow{T}_{i,j}$ and in $\overleftarrow{B}_{i,j}$ there exists a directed Manhattan path from
each point of $\overleftarrow{T}_{i,j}$ to $s_i\cup s_{j'}.$ 
\end{lemma}

\begin{proof}
By the algorithm, after Step 2 the supports of $\overrightarrow{A}_{i,j}$ and $\overleftarrow{B}_{i,j}$ are disjoint, however these networks are not yet feasible. Step 3 ensures feasibility of $\overrightarrow{A}_{i,j}$ (and $\overleftarrow{B}_{i,j}$) by
connecting each terminal $\overrightarrow{\bf{t}} \in \overrightarrow{T}_{i,j}$ to the network $\overrightarrow{A}_{i,j}$
computed in Step~2, using either a vertical segment $c$ (Fig. \ref{step2_3} Step~3 (a)) or a replacement path $\{b, c'\}$ (Fig. \ref{step2_3} Step~3 (b)). We will prove now that after Step~3
the networks $\overrightarrow{A}_{i,j}$ and $\overleftarrow{B}_{i,j}$ are disjoint and feasible.

By construction, $A_{i,j}$ connects each terminal of $\overrightarrow{T}_{i,j}$
to $s_{i'}\cup s_j$ by a Manhattan path. Therefore, it suffices to show that
this property is preserved each time when we orient a new part of the network
$A_{i,j}.$ This is obviously true when we orient all edges of
$\Gamma^{even}_{i,j}$ which have their support in $A_{i,j}.$ Now, suppose that a
segment $s$  added in $A_{i,j}$ to connect a terminal
$\overrightarrow{{\bf t}}$ is replaced by two directed segments $s'$ and $s''.$
Then all vertices of $\overrightarrow{T}_{i,j}$ connected in $A_{i,j}$ via $s$  (they are all located above $s$) will be now connected by directed Manhattan paths going via $s''$.
On the other hand, $\overrightarrow{{\bf t}}$ is connected in $\overrightarrow{A}_{i,j}$ via $c$ and $s'$ if $c$ does not belong to $\overleftarrow{B}_{i,j}$ and  via the replacement path $\{b,c'\}$ and $s'$ otherwise. We assert that in the last case, $c'$ and $b$
will not be used by $\overleftarrow{B}_{i,j}.$ Indeed, due to the pruning operation in Step 2, each $\overrightarrow{{\bf t}}$
is incident to at most one outgoing edge of $\overleftarrow{B}_{i,j}.$ Now, since we used a replacement path,  the segment
$c$ of $\Gamma^{odd}_{i,j}$ is still in $\overleftarrow{B}_{i,j}$ after the
pruning operation in Step 2.
This shows that $b$ cannot belong to $\overleftarrow{B_{i,j}}$. On the other
hand, the end-points of $c'$ do not belong to
$\overleftarrow{T}_{i,j},$ otherwise two distinct points of
$\overrightarrow{T}_{i,j} \cup \overleftarrow{T}_{i,j}$ will lie on the same vertical or horizontal line, which is impossible.
Now, since $c'$ belongs to $\Gamma^{even}_{i,j}$, by the algorithm, $c'$ can be
involved in a replacement path of $\overleftarrow{B}_{i,j}$ only if it is
incident to a point of $\overleftarrow{T}_{i,j}.$ Therefore $c'$ does not belong
to $\overleftarrow{B}_{i,j}$ either.
$\Box$
\end{proof}

\begin{lemma} \label{local}  $l(N_{i,j})=l(\overrightarrow{A}_{i,j} \cup
\overleftarrow{B}_{i,j})\le \textit{2 }l(N^*_{i,j})+l(N'_{i,j}),$ where
$N^*_{i,j}$ is an optimal completion  of $N'_{i,j}.$
\end{lemma}

\begin{proof}
By construction, the length of
$\overrightarrow{A}_{i,j} \cup \overleftarrow{B}_{i,j}$
is at most $\textit{2 }l(A_{i,j}) + \textit{2 }l(B_{i,j})$ plus
the total length $\Delta$ of the
 edges $c$ or $c'$ orthogonal to $s'$
defined in the algorithm.
Since any horizontal or vertical line crosses at most
one such edge, $\Delta$ is at most $l(N'_{i,j}),$ hence
$l(N_{i,j})\le \textit{2 }(l(A_{i,j}) + l(B_{i,j})) + l(N'_{i,j}).$
By optimality of $A_{i,j}$ and $B_{i,j},$
$l(A_{i,j})+l(B_{i,j})\le
l(\overrightarrow{N^*}_{i,j})+l(\overleftarrow{N^*}_{i,j})=l(N^*_{i,j}).$
$\Box$
\end{proof}

Let $N(T)$ be the network obtained as the union of $N'(T)$ and all local completions $N_{i,j}=\overrightarrow{A}_{i,j} \cup \overleftarrow{B}_{i,j}$ taken over all staircases. Let $N^*(T)$ be an optimal solution of
the BDMMN problem having $N'(T)$ as a subnetwork. Then $N^*(T)$ can be viewed as the disjoint union of $N'(T)$ with the local completions $N^*_{i,j}=(N^*(T)\setminus N'(T))\cap  {\cal S}_{i,j|i',j'}$ of $N'_{i,j}.$
It was shown in \cite{ChNouVa} that the interiors of two
staircases are disjoint. Since in the BDMMN problem there are no two terminals on the same horizontal or vertical line, two staircases cannot intersect in an edge, thus the intersection of two staircases is empty or a subset of terminals. Hence the local completions $N^*_{i,j}$ are pairwise disjoint (as well as the local completions $N_{i,j}$). By their definition, the networks $N'_{i,j}$ are also pairwise disjoint. Using this disjointness property, Lemma \ref{local}, and summing up over all staircases, we obtain that

\begin{eqnarray*}l(N(T)) &=& l(N'(T))+\sum l(N_{i,j})\le l(N'(T))+\sum
(\textit{2 }l(N^*_{i,j})+l(N'_{i,j}))\\
&\le& \textit{2 }l(N'(T))+\textit{2 }l(N^*(T)\setminus
N'(T))=\textit{2 }l(N^*(T)),
\end{eqnarray*}

The time complexity of the algorithm for the BDMMN problem
is dominated by the execution of the dynamic programming
algorithm that computes $A_{i,j}$ and $B_{i,j}$ for every
staircase ${\cal S}_{i,j|i',j'}.$ A staircase ${\cal S}_{i,j|i',j'}$ contributes
$O(|T_{i,j}|^3)$ to the total complexity of the algorithm.
Since each terminal belongs to at most two staircases,
the total complexity of the algorithm for the BDMMN
problem is $O(n^3),$ establishing the main result of the paper:

\begin{theorem} \label{2BDMMN} The network $N(T)$ computed by the algorithm
in $O(n^3)$ time is a factor 2 approximation for the BDMMN problem.
\end{theorem}

\noindent
{\bf Open question:}  {\it Is the BDMMN problem polynomial or NP-hard?} BDMMN
reduces only to staircases, avoiding thus the difficulty occurring in the MMN
problem
due to the interaction between strips and staircases. However, in the directed version,
for each staircase we have to compute two disjoint but not necessarily optimal directed
networks $\overrightarrow{N}_{i,j}$ and $\overleftarrow{N}_{i,j}$ which together provide
an optimal completion of $N'_{i,j}.$


\begin{thebibliography}{99}



\bibitem{BeWoWiSh} M. Benkert, A. Wolff, F. Widmann, and T. Shirabe,
The minimum Manhattan network problem: approximations and exact solutions,
{\it Comput. Geom.} {\bf 35} (2006), pp. 188--208.


\bibitem{CaChNoVa_normed} N. Catusse, V. Chepoi, K. Nouioua, and Y. Vax\`es,
Minimum Manhattan network problem in normed planes
with polygonal balls: a factor 2.5 approximation algorithm, Electronic preprint
arXiv:1004.5517v2, 2010.

\bibitem{ChNouVa} V. Chepoi, K. Nouioua, and Y. Vax\`es, A rounding algorithm for
approximating minimum Manhattan networks,
{\it Theor. Comput. Sci.} {\bf 390} (2008),  56--69 and {\it APPROX-RANDOM} 2005, pp. 40--51.

\bibitem{ChGuSu} F.Y.L. Chin, Z. Guo, and H. Sun, Minimum Manhattan network is NP-complete,
In {\it Symposium on Computational Geometry,} 2009, pp. 393--402.

\bibitem{ChAg} T.Y. Chung, D.P. Agrawal, On network characterization of an optimal broadcasting in the Manhattan Street Network, {\it INFOCOM}, 1990, pp. 465--472.



\bibitem{Epp} D. Eppstein, Spanning trees and spanners. In J.-R.
Sack and J. Urrutia, editors, Handbook of Computational Geometry,
pp. 425--461, Elsevier Science Publishers B.V. North-Holland,
Amsterdam, 2000.


\bibitem{FuSch} B. Fuchs and A. Schulze, A simple 3-approximation of minimum Manhattan networks, In {\it CTW} 2008,
26--29 (full version: Technical Report zaik2008-570).

\bibitem{GuSuZh1} Z. Guo, H. Sun, and H. Zhu, A fast 2-approximation algorithm for the minimum Manhattan network problem,
In {\it Proc. 4th International Conference on Algorithmic Aspects in
Information Management}, Lecture Notes in Computer Science vol.
5034, 2008, pp. 212--223.

\bibitem{GuSuZh} Z. Guo, H. Sun, and H. Zhu, Greedy construction of 2-approximation minimum Manhattan network,  In
{\it 19th International Symposium on Algorihtms and Computation,} Lecture Notes in Computer Science vol. 5369, 2008, pp. 4--15.

\bibitem{GuLeNa} J. Gudmundsson, C. Levcopoulos, and G. Narasimhan,
Approximating a minimum Manhattan network,
{\it Nordic J. Computing} {\bf 8} (2001) 219--232 and
{\it APPROX-RANDOM} 1999, pp. 28--37.


\bibitem{GuKlNaSmWo} J. Gudmundsson, R. Klein,  G. Narasimhan,  M. Smid, and A. Wolff, 06481 Abstracts Collection -- Geometric Networks and Metric Space Embeddings,
In {Geometric Networks and Metric Space Embeddings, Dagstuhl Seminar Proceedings,} 2007, Internationales Begegnungs- und Forschungszentrum f{\"u}r Informatik (IBFI),
Schloss Dagstuhl, Germany.


\bibitem{KaImAs} R. Kato, K. Imai, and T. Asano, An improved
algorithm for the minimum Manhattan network problem, In {\it 13th International Symposium on Algorihtms and Computation,}
Lecture Notes Computer Science, vol. 2518, 2002, pp. 344-356.

\bibitem{LaAlPa} F. Lam, M. Alexanderson, and L. Pachter, Picking
alignements from (Steiner) trees, {\it J. Comput. Biol.} {\bf 10}
(2003) 509--520.

\bibitem{Ma} N.F. Maxemchuk, Routing in the Manhattan Street Network, {\it IEEE Trans. Commun.} {\bf 35} (1987), 503--512.

\bibitem{NaSm} G. Narasimhan and M. Smid, {\it Geometric Spanner Networks,} Cambridge University Press,
2007.

\bibitem{Nou} K. Nouioua, {\it Enveloppes de Pareto et R\'eseaux de Manhattan:
caract\'erisations et algorithmes,} Th\`ese de Doctorat en
Informatique, Universit\'e de la M\'editerran\'ee, 2005.

\bibitem{RoThZw} L. Roditty, M. Thorup and U. Zwick, Roundtrip spanners and roundtrip routing in directed graphs, {\it SODA} 2002, pp 844--851.

\bibitem{SeUn} S. Seibert and W. Unger, A 1.5-approximation of the
minimal Manhattan network, In {\it 16th International Symposium on Algorihtms and Computation,}
Lecture Notes Computer Science,
vol. 3827, 2005, pp. 246--255.

\bibitem{Sch} A. Schulze, {\it Approximation Algorithms for Network
Design Problems,} Doctoral Thesis, Universit\"at zu K\"oln, 2008.

\bibitem{Va} E.A. Varvarigos, Optimal communication algorithms for Manhattan Street Networks, {\it Discrete Appl. Math.}, {\bf 83} (1998), 303--326.



\end{thebibliography}
\end{document}